\documentclass{llncs}
\usepackage{amssymb, amsmath,graphicx,graphics,caption,subcaption}
\usepackage[utf8]{inputenc} 
\usepackage[T1]{fontenc}
\usepackage{tikz}
\usepackage{hyperref}
\usepackage{tkz-graph}
\newcommand{\NP}{{\sf NP}}

\newcommand{\ssi}{\subseteq_i}

\newcommand{\li}{\subseteq_{li}}

\title{Classifying the Clique-Width\\ of $H$-Free Bipartite Graphs\thanks{The research in this paper was supported by EPSRC (EP/G043434/1 and EP/K025090/1) and ANR (TODO ANR-09-EMER-010).}} 
\author{Konrad K. Dabrowski\inst{1} \and Dani\"el Paulusma\inst{1}}
\institute{
School of Engineering and  Computing Sciences, Durham University,\\
Science Laboratories, South Road,
Durham DH1 3LE, United Kingdom
\texttt{\{konrad.dabrowski,daniel.paulusma\}@durham.ac.uk}
}

\begin{document}
\maketitle
\begin{abstract}
Let $G$ be a bipartite graph, and let $H$ be a bipartite graph with a fixed bipartition $(B_H,W_H)$.
We consider three different, natural ways of forbidding $H$ as an induced subgraph in $G$. 
First, $G$ is $H$-free if it does not contain $H$ as an induced subgraph.
Second,  $G$ is strongly $H$-free if $G$ is $H$-free or else has no bipartition $(B_G,W_G)$ 
with  $B_H\subseteq B_G$ and $W_H\subseteq W_G$.
Third, $G$ is weakly $H$-free if $G$ is $H$-free or else has at least one bipartition $(B_G,W_G)$ 
with $B_H\not\subseteq B_G$ or $W_H\not\subseteq W_G$.
Lozin and Volz characterized all bipartite graphs $H$ for which the class of strongly $H$-free bipartite graphs 
has bounded clique-width. 
We extend their result by giving complete classifications for the other two variants of $H$-freeness.
\end{abstract}

\section{Introduction}

The {\em clique-width} of a graph $G$, is a well-known graph parameter that has been studied both in a structural and in an algorithmic context. It is the minimum 
number of labels needed to construct $G$ by
using the following four operations:
\begin{enumerate}
\item[(i)] creating a new graph consisting of a single vertex $v$ with label $i$;
\item[(ii)] taking the disjoint union of two labelled graphs $G_1$ and $G_2$;
\item[(iii)] joining each vertex with label $i$ to each vertex with label $j$ ($i\neq j$);
\item[(iv)] renaming label $i$ to $j$.
\end{enumerate}
We refer to the surveys of Gurski~\cite{Gu07} and Kami\'nski, Lozin and Milani\v{c} \cite{KLM09}
for an in-depth study of the properties of clique-width. 

We say that a class of graphs has {\em bounded} clique-width if every graph from the class has clique-width at most $p$ for some constant $p$.
As many \NP-hard graph problems can be solved in polynomial time on graph classes of bounded
clique-width~\cite{CMR00,KR03b,Oum08,Ra07}, it is natural to determine whether a certain graph class has bounded clique-width and to find
new graph classes of bounded clique-width. In particular, many papers determined the clique-width of graph classes characterized by one or more forbidden induced subgraphs~\cite{BL02,BELL06,BGMS14,BKM06,BK05,BLM04b,BLM04,BM02,BM03,DGP13,LR04,LR06,LV08,MR99}.

In this paper we focus on classes of bipartite graphs characterized by a forbidden induced subgraph $H$. A 
graph $G$ is $H$-free if it does not contain $H$ as an induced subgraph. If $G$ is bipartite, then when considering notions for $H$-freeness, we may assume without loss of generality that $H$ is bipartite as well.
For bipartite graphs, the situation is more subtle as one can define the notion of freeness with respect to a fixed bipartition $(B_H,W_H)$ of~$H$. 
This leads to two other notions (also see Section~\ref{s-preliminaries} for formal definitions). We say that a bipartite graph $G$ is strongly $H$-free if $G$ is $H$-free or else has no bipartition $(B_G,W_G)$ 
with  $B_H\subseteq B_G$ and $W_H\subseteq W_G$.
Strongly $H$-free graphs have been studied with respect to their clique-width, although under less explicit  terminology (see e.g. \cite{KLM09,LR06,LV08}).
In particular, Lozin and Volz~\cite{LV08} completely determined those bipartite graphs $H$, for which the class of strongly  $H$-free graphs has bounded clique-width (we give an exact statement of their result in Section~\ref{s-main}).
If $G$ is $H$-free or else has at least one bipartition $(B_G,W_G)$  with $B_H\not\subseteq B_G$ or $W_H\not\subseteq W_G$, then $G$ is said to be weakly $H$-free. As far as we are aware this notion has not been studied with respect to the clique-width of bipartite graphs.

\medskip
\noindent
{\bf Our Results:}
We completely classify the classes of $H$-free bipartite graphs of bounded clique-width. 
We also introduce the notion of weakly $H$-freeness for bipartite graphs
and characterize those classes of weakly $H$-free bipartite graphs that have bounded clique-width.
In this way, we have identified a number of new graph classes of bounded clique-width.
Before stating our results precisely in Section~\ref{s-main}, we first give some terminology and examples in Section~\ref{s-preliminaries}.
In Section~\ref{s-proofs} we give the proofs of our results.

\section{Terminology and Examples}\label{s-preliminaries}

We first give some terminology on general graphs, followed by terminology for bipartite graphs. We illustrate the definitions of $H$-freeness,
strong $H$-freeness and weak $H$-freeness of bipartite graphs with some examples. 
As we will explain, these examples also make clear that all three notions are different from each other.

\medskip
\noindent
{\bf General graphs:} Let $G$ and $H$ be graphs.
We write $H\ssi G$ to indicate that $H$ is an induced subgraph of $G$.
A bijection of the vertices $f:V_G\to V_H$ is called a ({\it graph}) {\it isomorphism} when $uv\in E_G$ if and only if $f(u)f(v)\in E_H$.
If such a bijection exists then $G$ and $H$ are {\it isomorphic}. 
Let  $\{H_1,\ldots,H_p\}$ be a set of graphs. 
A graph $G$ is {\it $(H_1,\ldots,H_p)$-free} if no $H_i$ is an induced subgraph of $G$.
If $p=1$ we may write $H_1$-free instead of $(H_1)$-free.
The {\it disjoint union} $G+H$ of two vertex-disjoint graphs $G$ and $H$ is the graph with vertex set $V_G\cup V_H$ and edge set $E_G\cup E_H$. We denote the disjoint union of $r$ copies of $G$ by $rG$.

\medskip
\noindent
{\bf Bipartite graphs:} A graph $G$ is {\em bipartite} if its vertex set can be partitioned into two (possibly empty) independent sets.
Let $H$ be a bipartite graph.
We say that $H$ is a \emph{labelled} bipartite graph if we are also given a \emph{black-and-white labelling} $\ell$, which is a labelling 
that assigns either the colour ``black'' or the colour ``white''  to each vertex of $H$ in such a way that the
two resulting monochromatic colour classes $B_H^\ell$ and $W_H^\ell$ form a partition of $H$ into two (possibly empty) independent sets. 
From now on we denote a graph $H$ with such a labelling $\ell$ by $H^\ell=(B_H^\ell,W_H^\ell,E_H)$. 
Here the pair $(B_H^\ell,W_H^\ell)$ is {\it ordered}, that is, $(B_H^\ell,W_H^\ell,E_H)$ and $(W_H^\ell,B_H^\ell,E_H)$ are different labelled 
bipartite graphs.

We say that two labelled bipartite graphs $H_1^\ell$ and $H_2^{\ell^*}$ are \emph{isomorphic} if  the (unlabelled) graphs $H_1$ and $H_2$ are isomorphic, and if in addition there exists an isomorphism $f: V_{H_1}\to V_{H_2}$ such that for all $u\in V_{H_1}$,
$u\in W^\ell_{H_1}$ if and only if $f(u)\in W^{\ell^*}_{H_2}$.
Moreover, if $H_1=H_2$, then $\ell$ and $\ell^*$ are said to be {\em isomorphic} labellings.
For example, the bipartite graphs $(\{u,v\},\emptyset)$ and $(\{x,y\},\emptyset)$ are isomorphic, and the labelled bipartite graph $(\{u,v\},\emptyset,\emptyset)$
is isomorphic to the labelled bipartite graph $(\{x,y\},\emptyset, \emptyset)$.
However, $(\{x,y\},\emptyset,\emptyset)$ is neither isomorphic to $(\emptyset,\{x,y\},\emptyset)$
nor to $(\{x\},\{y\},\emptyset)$ (also see Figure~\ref{f-2p1}).

We write $H_1^\ell \li H_2^{\ell^*}$ if $H_1\ssi H_2$, $B_{H_1}^\ell\subseteq B_{H_2}^{\ell^*}$ and $W_{H_1}^\ell\subseteq W_{H_2}^{\ell^*}$.
In this case we say that $H_1^\ell$ is a \emph{labelled} induced subgraph of $H_2^{\ell^*}$. 
Note that the two labelled bipartite graphs $H_1^{\ell_1}$ and $H_2^{\ell_2}$ are isomorphic if and only if $H_1^{\ell_1}$ is a labelled induced subgraph of
$H_2^{\ell_2}$, and vice versa. 

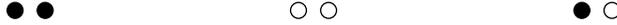
\begin{figure}
\begin{center}
\begin{subfigure}{0.3\textwidth}
\centering
\begin{tikzpicture}[scale=0.4]
\GraphInit[vstyle=Simple]
\SetVertexSimple[MinSize=6pt]
\Vertex[x=0,y=0]{x}
\Vertex[x=1,y=0]{y}
\end{tikzpicture}
\end{subfigure}
\begin{subfigure}{0.3\textwidth}
\centering
\begin{tikzpicture}[scale=0.4]
\GraphInit[vstyle=Simple]
\SetVertexSimple[MinSize=6pt,FillColor=white]
\Vertex[x=0,y=0]{x}
\Vertex[x=1,y=0]{y}
\end{tikzpicture}
\end{subfigure}
\begin{subfigure}{0.3\textwidth}
\centering
\begin{tikzpicture}[scale=0.4]
\GraphInit[vstyle=Simple]
\SetVertexSimple[MinSize=6pt]
\Vertex[x=0,y=0]{x}
\SetVertexSimple[MinSize=6pt,FillColor=white]
\Vertex[x=1,y=0]{y}
\end{tikzpicture}
\end{subfigure}
\caption{The graph $2P_1$ partitioned into three ways; none of these three labelled bipartite graphs are isomorphic to each other.}\label{f-2p1}
\end{center}
\end{figure}

Let $G$ be an (unlabelled) bipartite graph, and let $H^\ell$ be a labelled bipartite graph. 
We say that $G$ contains $H^\ell$ as a \emph{strongly labelled} induced subgraph if $H^\ell\li (B_G,W_G,E_G)$ for some bipartition $(B_G,W_G)$ of $G$.
If not, then $G$ is said to be \emph{strongly $H^\ell$-free}.
We say that $G$ contains $H^\ell$ as a \emph{weakly labelled} induced subgraph if $H^\ell\li (B_G,W_G,E_G)$ for all bipartitions $(B_G,W_G)$ of $G$.
If not, then $G$ is said to be \emph{weakly $H^\ell$-free}. Equivalently, $G$
is strongly $H^\ell$-free if for every labelling $\ell^*$ of $G$, $G^{\ell^*}$
does not contain $H^\ell$ as a labelled induced subgraph and $G$ is weakly
$H^\ell$-free if there is a labelling $\ell^*$ of $G$ such that $G^{\ell^*}$
does not contain $H^\ell$ as a labelled induced subgraph.
Note that these two notions of freeness are only defined for (unlabelled) bipartite graphs. 
Let $\{H_1^{\ell_1},\ldots, H_p^{\ell_p}\}$ be a set of labelled bipartite graphs. Then a graph $G$ is {\it strongly (weakly) $(H_1^{\ell_1},\ldots, H_p^{\ell_p})$-free} if $G$ is 
strongly (weakly) $H_i^{\ell_i}$-free for $i=1,\ldots,p$.
 
The following lemma shows that for all labelled bipartite graphs $H^\ell$, the class of $H$-free graphs 
is a (possibly proper) subclass of the class of strongly $H^\ell$-free bipartite graphs and that the latter graph class is a (possibly proper) subclass of the class of weakly $H^\ell$-free bipartite graphs.
 
\begin{lemma}\label{l-observation}
Let $G$ be a bipartite graph and $H^\ell$ be a labelled bipartite graph. The following two statements hold:
\begin{itemize}
\item [(i)] If $G$ is $H$-free, then $G$ is strongly $H^\ell$-free.
\item [(ii)] If $G$ is strongly $H^\ell$-free, then $G$ is weakly $H^\ell$-free.
\end{itemize}
Moreover, the two reverse statements are not necessarily true.
\end{lemma} 

\begin{proof}
Statements (i)~and~(ii) follow by definition. The following two examples, which are also depicted in Figure~\ref{f-graphsfromobservation}, show that the reverse statements may not necessarily be true.
Let  $G$ be isomorphic to $S_{1,1,3}$ with
$V_{G}=\{u_1,\ldots,u_6\}$ and $E_{G}=\{u_1u_2,u_1u_3,u_1u_4,u_4u_5,u_5u_6\}$.
Let $H=K_{1,3}+P_1$. We denote the vertex set and edge set of $H$ by $V_H=\{x_1,x_2,x_3,x_4,x_5\}$ and $E_H=\{x_1x_2,x_1x_3,x_1x_4\}$.

Let $H^\ell=(\{x_2,x_3,x_4\},\{x_1,x_5\},E_H)$.
We first notice that $G$ is not $H$-free, because $G[u_1,u_2,u_3,u_4,u_6]$ is isomorphic to $K_{1,3}+P_1$.
However, we do have that $G$ is strongly $H^\ell$-free, because $H^\ell$ is neither a labelled induced subgraph of $(\{u_1,u_5\},\{u_2,u_3,u_4,u_6\},E_G\}$ nor of
$(\{u_2,u_3,u_4,u_6\},\{u_1,u_5\},E_G\}$.

Let $H^{\ell^*}=(\{x_2,x_3,x_4,x_5\},\{x_1\},E_H)$. Then $G$ is not strongly $H^{\ell^*}$-free, because
$(\{u_2,u_3,u_4,u_6\},\{u_1\},\{u_1u_2,u_1u_3,u_1u_4\})$ is isomorphic to $H^{\ell^*}$.  
However, $G$ is weakly $H^{\ell^*}$-free, because $H^{\ell^*}$ is not a labelled induced subgraph of  $(\{u_1,u_5\},\{u_2,u_3,u_4,u_6\},E_G\})$.
\qed
\end{proof}

\begin{figure}
\begin{center}
\subcaptionbox{$G$}[0.3\textwidth]
{\centering
\begin{tikzpicture}[scale=1]
\GraphInit[vstyle=Empty]
\Vertex[x=0,y=0,L=$u_1$]{x}
\Vertex[x=1,y=0,L=$u_4$]{x1}
\Vertex[x=2,y=0,L=$u_5$]{x2}
\Vertex[x=3,y=0,L=$u_6$]{x3}
\Vertex[a=120,d=1,L=$u_2$]{y1}
\Vertex[a=-120,d=1,L=$u_3$]{y2}
\Edges(y1,x,y2)
\Edges(x,x1,x2,x3)
\end{tikzpicture}
}
\subcaptionbox{$H^\ell$}[0.38\textwidth]
{\centering
\centering
\begin{tikzpicture}[scale=1]
\GraphInit[vstyle=Simple]
\SetVertexSimple[MinSize=6pt]
\Vertex[x=-0.5,y=0]{y1}
\Vertex[x=0,y=0]{y2}
\Vertex[x=0.5,y=0]{y3}
\SetVertexSimple[MinSize=6pt,FillColor=white]
\Vertex[x=0,y=1]{x}
\Vertex[x=1,y=1]{z}
\Edge(x)(y1)
\Edge(x)(y2)
\Edge(x)(y3)
\end{tikzpicture}
}
\subcaptionbox{$H^{\ell^*}$}[0.2\textwidth]
{\centering
\centering
\begin{tikzpicture}[scale=1]
\GraphInit[vstyle=Simple]
\SetVertexSimple[MinSize=6pt]
\Vertex[x=1,y=0]{z}
\Vertex[x=-0.5,y=0]{y1}
\Vertex[x=0,y=0]{y2}
\Vertex[x=0.5,y=0]{y3}
\SetVertexSimple[MinSize=6pt,FillColor=white]
\Vertex[x=0,y=1]{x}
\Edge(x)(y1)
\Edge(x)(y2)
\Edge(x)(y3)
\end{tikzpicture}
}
\caption{The graphs $G, H^\ell$ and $H^{\ell^*}$ from the proof of Lemma~\ref{l-observation}.}
\label{f-graphsfromobservation}
\end{center}
\end{figure}
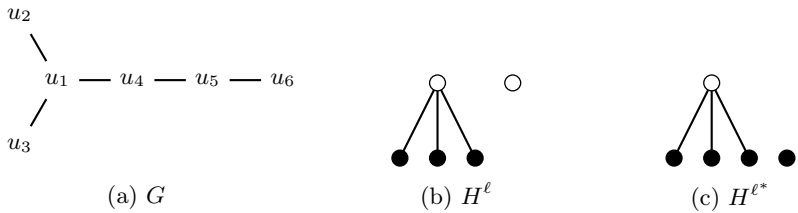
\vspace*{-0.8cm}

\noindent
{\bf Special Graphs:}
For $r\geq 1$, the graphs $C_r$, $K_r$, $P_r$ denote the cycle, complete graph and path on $r$ vertices, respectively, and the graph $K_{1,r}$ denotes 
the star on $r+1$ vertices. If $r=3$, the graph $K_{1,r}$ is also called the {\em claw}.
For $1\leq h\leq i\leq j$, let $S_{h,i,j}$ denote the tree that has only one vertex $x$ of degree $3$ and that has exactly three leaves, which are of distance $h$, $i$ and $j$ from $x$, respectively.
Observe  that $S_{1,1,1}=K_{1,3}$. A graph $S_{h,i,j}$ is called a {\it subdivided claw}.

Let $H^\ell=(B_H^\ell,W_H^\ell,E_H)$ be a labelled bipartite graph.
The {\em opposite} of  $H^\ell$ is defined as 
the labelled bipartite graph $H^{\overline{\ell}}=(W_H^\ell,B_H^\ell,E_H)$. We say that $\overline{\ell}$ is the {\em opposite} black-and-white labelling of~$\ell$. Suppose that $H$ is a bipartite graph such that among all its black-and-white labellings, all those that maximize the number of black vertices are isomorphic.
In this case we pick one of such labelling and call 
it $b$.

\section{The Classifications}\label{s-main}

A full classification of the boundedness of the clique-width of strongly $H^\ell$-free
bipartite graphs was given by Lozin and Voltz~\cite{LV08} except that in their result the trivial case when $H^\ell = (sP_1)^b$ or $H^\ell = (sP_1)^{\overline{b}}$ for some $s\geq 1$ was missing. Their proof is correct except it overlooked this case, which occurs when one of the colour classes of a labelled graph $H^\ell$ is empty. However, strongly $(sP_1)^b$-free bipartite graphs can have at most $2s-2$ vertices, and as such form a class of bounded clique-width.
Below we state their result after incorporating this small correction, followed by our results for the other two variants of freeness.
We refer to Figure~\ref{f-graphsfromtheorems} for pictures of the labelled bipartite graphs used in Theorems~\ref{thm:bipvadim} and~\ref{thm:bipweak}.

\begin{theorem}[\cite{LV08}]\label{thm:bipvadim}
Let $H^\ell$  be a  labelled bipartite graph. The class of strongly $H^\ell$-free bipartite graphs has bounded clique-width if and only if
one of the following cases holds:
\begin{itemize}
\item [$\bullet$] $H^\ell = (sP_1)^b$ \hspace*{15.7mm} or \hspace*{1.5mm} $H^\ell = (sP_1)^{\overline{b}}$ \hspace*{17mm} for some $s\geq 1$
\item [$\bullet$]  $H^\ell \li (K_{1,3}+3P_1)^b$  \hspace*{2mm}  or \hspace*{1.5mm} $H^\ell \li (K_{1,3}+3P_1)^{\overline{b}}$
\item [$\bullet$] $H^\ell\li (K_{1,3}+P_2)^{b}$   \hspace*{3.5mm}  or \hspace*{1.5mm} $H^\ell\li (K_{1,3}+P_2)^{\overline{b}}$
 \item [$\bullet$] $H^\ell \li  (P_1+S_{1,1,3})^b$  \hspace*{2mm}  or \hspace*{1.5mm}  $H^\ell \li  (P_1+S_{1,1,3})^{\overline{b}}$
\item [$\bullet$] $H^\ell \li (S_{1,2,3})^{b}$   \hspace*{9.8mm}  or  \hspace*{1.5mm} $H^\ell \li (S_{1,2,3})^{\overline{b}}$.
\end{itemize}
\end{theorem}

\begin{theorem}\label{thm:bipartite}
Let $H$ be a graph. The class of $H$-free bipartite graphs has bounded
clique-width if and only if one of the following cases holds:
\begin{itemize}
\item [$\bullet$] $H=sP_1$ for some $s\geq 1$
\item [$\bullet$] $H\ssi K_{1,3}+3P_1$
\item [$\bullet$] $H\ssi K_{1,3}+P_2$ 
\item [$\bullet$] $H\ssi P_1+S_{1,1,3}$
\item [$\bullet$] $H\ssi S_{1,2,3}$.
\end{itemize}
\end{theorem}

\begin{theorem}\label{thm:bipweak}
Let $H^\ell$ be a labelled bipartite graph.
The class of weakly $H^\ell$-free bipartite graphs has bounded clique-width if and only if one of the following cases holds:
\begin{itemize}
\item [$\bullet$] $H^\ell \hspace*{1mm} = (sP_1)^b$  \hspace*{11.7mm} or  \hspace*{1.5mm}  $H^\ell = (sP_1)^{\overline{b}}$  \hspace*{17mm}  for some $s\geq 1$
\item [$\bullet$] $H^\ell\li (2P_1+P_3)^b$   \hspace*{2mm}  or \hspace*{1.5mm} $H^\ell\li (2P_1+P_3)^{\overline{b}}$
\item [$\bullet$] $H^\ell\li (P_1+P_5)^b$ \hspace*{3.3mm}  or \hspace*{1.5mm} $H^\ell\li (P_1+P_5)^{\overline{b}}$
\item [$\bullet$]  $H \hspace*{1mm} \ssi P_2+P_4$
\item [$\bullet$] $H \hspace*{1mm} \ssi P_6$.
\end{itemize}
\end{theorem}

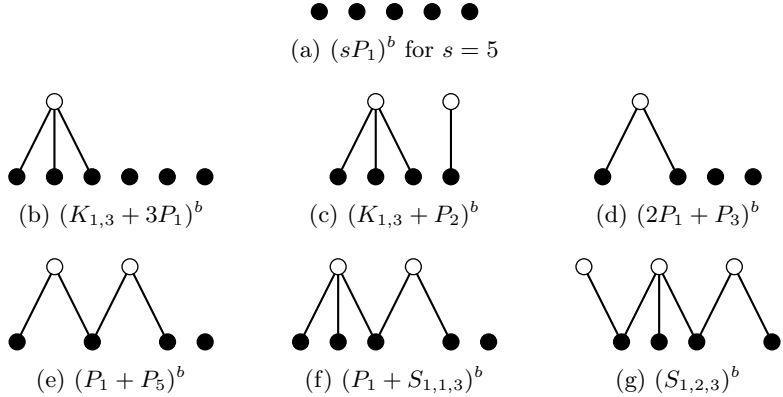
\begin{figure}
\begin{center}
\begin{subfigure}{0.3\textwidth}
\centering
\begin{tikzpicture}[scale=1]
\GraphInit[vstyle=Simple]
\SetVertexSimple[MinSize=6pt]
\Vertex[x=0,y=0]{x0}
\Vertex[x=0.5,y=0]{x1}
\Vertex[x=1,y=0]{x2}
\Vertex[x=1.5,y=0]{x3}
\Vertex[x=2,y=0]{x4}
\end{tikzpicture}
\subcaption{$(sP_1)^b$ for $s=5$}
\end{subfigure}

\vspace{1em}
\begin{subfigure}{0.3\textwidth}
\centering
\begin{tikzpicture}[scale=1]
\GraphInit[vstyle=Simple]
\SetVertexSimple[MinSize=6pt]
\Vertex[x=0.5,y=0]{x0}
\Vertex[x=1,y=0]{x1}
\Vertex[x=1.5,y=0]{x2}
\Vertex[x=2,y=0]{x3}
\Vertex[x=2.5,y=0]{x4}
\Vertex[x=3,y=0]{x5}
\SetVertexSimple[MinSize=6pt,FillColor=white]
\Vertex[x=1,y=1]{z}
\Edge(x0)(z)
\Edge(x1)(z)
\Edge(x2)(z)
\end{tikzpicture}
\subcaption{$(K_{1,3}+3P_1)^b$}
\end{subfigure}
\begin{subfigure}{0.3\textwidth}
\centering
\begin{tikzpicture}[scale=1]
\GraphInit[vstyle=Simple]
\SetVertexSimple[MinSize=6pt]
\Vertex[x=0.5,y=0]{x0}
\Vertex[x=1,y=0]{x1}
\Vertex[x=1.5,y=0]{x2}
\Vertex[x=2,y=0]{x3}
\SetVertexSimple[MinSize=6pt,FillColor=white]
\Vertex[x=1,y=1]{z}
\Vertex[x=2,y=1]{y}
\Edge(x0)(z)
\Edge(x1)(z)
\Edge(x2)(z)
\Edge(x3)(y)
\end{tikzpicture}
\subcaption{$(K_{1,3}+P_2)^b$}
\end{subfigure}
\begin{subfigure}{0.3\textwidth}
\centering
\begin{tikzpicture}[scale=1]
\GraphInit[vstyle=Simple]
\SetVertexSimple[MinSize=6pt]
\Vertex[x=0,y=0]{x0}
\Vertex[x=1,y=0]{x2}
\Vertex[x=1.5,y=0]{x4}
\Vertex[x=2,y=0]{x5}
\SetVertexSimple[MinSize=6pt,FillColor=white]
\Vertex[x=0.5,y=1]{x1}
\Edges(x0,x1,x2)
\end{tikzpicture}
\subcaption{$(2P_1+P_3)^b$}
\end{subfigure}

\vspace{1em}
\begin{subfigure}{0.3\textwidth}
\centering
\begin{tikzpicture}[scale=1]
\GraphInit[vstyle=Simple]
\SetVertexSimple[MinSize=6pt]
\Vertex[x=0,y=0]{x0}
\Vertex[x=1,y=0]{x2}
\Vertex[x=2,y=0]{x4}
\Vertex[x=2.5,y=0]{x5}
\SetVertexSimple[MinSize=6pt,FillColor=white]
\Vertex[x=0.5,y=1]{x1}
\Vertex[x=1.5,y=1]{x3}
\Edges(x0,x1,x2,x3,x4)
\end{tikzpicture}
\subcaption{$(P_1+P_5)^b$}
\end{subfigure}
\begin{subfigure}{0.3\textwidth}
\centering
\begin{tikzpicture}[scale=1]
\GraphInit[vstyle=Simple]
\SetVertexSimple[MinSize=6pt]
\Vertex[x=0.5,y=0]{x0}
\Vertex[x=1,y=0]{x1}
\Vertex[x=1.5,y=0]{x2}
\Vertex[x=2.5,y=0]{x4}
\Vertex[x=3,y=0]{x5}
\SetVertexSimple[MinSize=6pt,FillColor=white]
\Vertex[x=1,y=1]{z}
\Vertex[x=2,y=1]{x3}
\Edge(x0)(z)
\Edge(x1)(z)
\Edge(x2)(z)
\Edges(x2,x3,x4)
\end{tikzpicture}
\subcaption{$(P_1+S_{1,1,3})^b$}
\end{subfigure}
\begin{subfigure}{0.3\textwidth}
\centering
\begin{tikzpicture}[scale=1]
\GraphInit[vstyle=Simple]
\SetVertexSimple[MinSize=6pt]
\Vertex[x=0.5,y=0]{x0}
\Vertex[x=1,y=0]{x1}
\Vertex[x=1.5,y=0]{x2}
\Vertex[x=2.5,y=0]{x4}
\SetVertexSimple[MinSize=6pt,FillColor=white]
\Vertex[x=1,y=1]{z}
\Vertex[x=2,y=1]{x3}
\Vertex[x=0,y=1]{x5}
\Edges(x5,x0,z)
\Edge(x1)(z)
\Edge(x2)(z)
\Edges(x2,x3,x4)
\end{tikzpicture}
\subcaption{$(S_{1,2,3})^b$}
\end{subfigure}
\caption{The labelled bipartite graphs used in Theorems~\ref{thm:bipvadim} and~\ref{thm:bipweak}.}
\label{f-graphsfromtheorems}
\end{center}
\end{figure}

\section{The Proofs of Our Results}\label{s-proofs}

We first recall a number of basic facts on clique-width known from the literature. We then state a number of other lemmas which we use to prove 
Theorems~\ref{thm:bipartite} and~\ref{thm:bipweak}.

\subsection{Facts about Clique-width}\label{s-facts}

The {\em bipartite complement} of a bipartite graph {\em with respect to} a bipartition $(B,W)$ is the bipartite graph with bipartition $(B,W)$, in which two vertices $u\in B$ and $v\in W$  are adjacent if and only if $uv\notin E$. 
For instance, the graph $2P_2$ has $C_4$ as its only bipartite complement, whereas the graph
$2P_1$ has $2P_1$ and $P_2$ as its bipartite complements.
For two disjoint vertex subsets  $X$ and $Y$ in $G$, the {\em bipartite complementation} operation with respect to $X$ and $Y$ acts on $G$ by replacing every edge with one end-vertex in $X$ and the other one in $Y$ by a non-edge and vice versa. 
The \emph{edge subdivision} operation replaces an edge $vw$ in a graph by a new vertex $u$ with edges $uv$ and $uw$.

We now state some useful facts for dealing with clique-width. We will use these facts throughout the paper.
We will say that a graph operation {\em preserves} boundedness of clique-width if 
for any constant $k$ and any graph class ${\cal G}$, the graph class 
${\cal G}_{[k]}$ obtained by performing the operation at most $k$ times on each graph in ${\cal G}$
has bounded clique-width if and only if  ${\cal G}$  has bounded clique-width.

\medskip
\noindent
{\it Fact 1.} Vertex deletion preserves boundedness of clique-width~\cite{LR04}.

\medskip
\noindent
{\it Fact 2.} Bipartite complementation preserves boundedness of clique-width~\cite{KLM09}.

\medskip
\noindent
{\it Fact 3.}
For a class of graphs ${\cal G}$ of {\em bounded} degree, let ${\cal G}'$ be the class of graphs obtained from ${\cal G}$ by applying zero or more  
edge subdivisions operations to each graph in ${\cal G}$. 
Then ${\cal G}$ has bounded clique-width if and only if  ${\cal G}'$ has bounded clique-width~\cite{KLM09}.

\medskip
\noindent
We also use some other elementary results on the clique-width of graphs.
In order to do so we need the notion of a {\em wall}.
We do not formally define this notion, but instead refer to Figure~\ref{f-walls}, in which three examples of walls of different height are depicted. A \emph{$k$-subdivided wall} is a graph obtained from a wall after subdividing each edge 
exactly $k$ times for some constant $k\geq 0$.
 The next  well-known lemma follows from combining Fact~3 with the fact that walls 
 have maximum degree~3  and unbounded clique-width (see e.g.~\cite{KLM09}).

\begin{lemma}\label{l-walls}
For any constant $k$, the class of $k$-subdivided walls has unbounded clique-width.
\end{lemma}

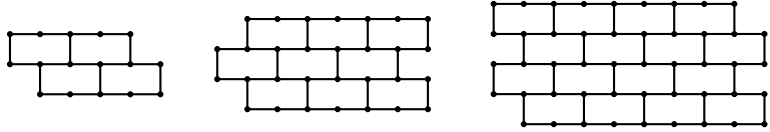
\begin{figure}
\begin{center}
\begin{subfigure}{0.2\textwidth}
\centering
\begin{tikzpicture}[scale=0.4, every node/.style={scale=0.3}]
\GraphInit[vstyle=Simple]
\SetVertexSimple[MinSize=6pt]
\Vertex[x=1,y=0]{v10}
\Vertex[x=2,y=0]{v20}
\Vertex[x=3,y=0]{v30}
\Vertex[x=4,y=0]{v40}
\Vertex[x=5,y=0]{v50}

\Vertex[x=0,y=1]{v01}
\Vertex[x=1,y=1]{v11}
\Vertex[x=2,y=1]{v21}
\Vertex[x=3,y=1]{v31}
\Vertex[x=4,y=1]{v41}
\Vertex[x=5,y=1]{v51}

\Vertex[x=0,y=2]{v02}
\Vertex[x=1,y=2]{v12}
\Vertex[x=2,y=2]{v22}
\Vertex[x=3,y=2]{v32}
\Vertex[x=4,y=2]{v42}

\Edges(    v10,v20,v30,v40,v50)
\Edges(v01,v11,v21,v31,v41,v51)
\Edges(v02,v12,v22,v32,v42)

\Edge(v01)(v02)

\Edge(v10)(v11)

\Edge(v21)(v22)

\Edge(v30)(v31)

\Edge(v41)(v42)

\Edge(v50)(v51)

\end{tikzpicture}
\end{subfigure}
\begin{subfigure}{0.3\textwidth}
\centering
\begin{tikzpicture}[scale=0.4, every node/.style={scale=0.3}]
\GraphInit[vstyle=Simple]
\SetVertexSimple[MinSize=6pt]
\Vertex[x=1,y=0]{v10}
\Vertex[x=2,y=0]{v20}
\Vertex[x=3,y=0]{v30}
\Vertex[x=4,y=0]{v40}
\Vertex[x=5,y=0]{v50}
\Vertex[x=6,y=0]{v60}
\Vertex[x=7,y=0]{v70}

\Vertex[x=0,y=1]{v01}
\Vertex[x=1,y=1]{v11}
\Vertex[x=2,y=1]{v21}
\Vertex[x=3,y=1]{v31}
\Vertex[x=4,y=1]{v41}
\Vertex[x=5,y=1]{v51}
\Vertex[x=6,y=1]{v61}
\Vertex[x=7,y=1]{v71}

\Vertex[x=0,y=2]{v02}
\Vertex[x=1,y=2]{v12}
\Vertex[x=2,y=2]{v22}
\Vertex[x=3,y=2]{v32}
\Vertex[x=4,y=2]{v42}
\Vertex[x=5,y=2]{v52}
\Vertex[x=6,y=2]{v62}
\Vertex[x=7,y=2]{v72}

\Vertex[x=1,y=3]{v13}
\Vertex[x=2,y=3]{v23}
\Vertex[x=3,y=3]{v33}
\Vertex[x=4,y=3]{v43}
\Vertex[x=5,y=3]{v53}
\Vertex[x=6,y=3]{v63}
\Vertex[x=7,y=3]{v73}

\Edges(    v10,v20,v30,v40,v50,v60,v70)
\Edges(v01,v11,v21,v31,v41,v51,v61,v71)
\Edges(v02,v12,v22,v32,v42,v52,v62,v72)
\Edges(    v13,v23,v33,v43,v53,v63,v73)

\Edge(v01)(v02)

\Edge(v10)(v11)
\Edge(v12)(v13)

\Edge(v21)(v22)

\Edge(v30)(v31)
\Edge(v32)(v33)

\Edge(v41)(v42)

\Edge(v50)(v51)
\Edge(v52)(v53)

\Edge(v61)(v62)

\Edge(v70)(v71)
\Edge(v72)(v73)
\end{tikzpicture}
\end{subfigure}
\begin{subfigure}{0.35\textwidth}
\centering
\begin{tikzpicture}[scale=0.4, every node/.style={scale=0.3}]
\GraphInit[vstyle=Simple]
\SetVertexSimple[MinSize=6pt]
\Vertex[x=1,y=0]{v10}
\Vertex[x=2,y=0]{v20}
\Vertex[x=3,y=0]{v30}
\Vertex[x=4,y=0]{v40}
\Vertex[x=5,y=0]{v50}
\Vertex[x=6,y=0]{v60}
\Vertex[x=7,y=0]{v70}
\Vertex[x=8,y=0]{v80}
\Vertex[x=9,y=0]{v90}

\Vertex[x=0,y=1]{v01}
\Vertex[x=1,y=1]{v11}
\Vertex[x=2,y=1]{v21}
\Vertex[x=3,y=1]{v31}
\Vertex[x=4,y=1]{v41}
\Vertex[x=5,y=1]{v51}
\Vertex[x=6,y=1]{v61}
\Vertex[x=7,y=1]{v71}
\Vertex[x=8,y=1]{v81}
\Vertex[x=9,y=1]{v91}

\Vertex[x=0,y=2]{v02}
\Vertex[x=1,y=2]{v12}
\Vertex[x=2,y=2]{v22}
\Vertex[x=3,y=2]{v32}
\Vertex[x=4,y=2]{v42}
\Vertex[x=5,y=2]{v52}
\Vertex[x=6,y=2]{v62}
\Vertex[x=7,y=2]{v72}
\Vertex[x=8,y=2]{v82}
\Vertex[x=9,y=2]{v92}

\Vertex[x=0,y=3]{v03}
\Vertex[x=1,y=3]{v13}
\Vertex[x=2,y=3]{v23}
\Vertex[x=3,y=3]{v33}
\Vertex[x=4,y=3]{v43}
\Vertex[x=5,y=3]{v53}
\Vertex[x=6,y=3]{v63}
\Vertex[x=7,y=3]{v73}
\Vertex[x=8,y=3]{v83}
\Vertex[x=9,y=3]{v93}

\Vertex[x=0,y=4]{v04}
\Vertex[x=1,y=4]{v14}
\Vertex[x=2,y=4]{v24}
\Vertex[x=3,y=4]{v34}
\Vertex[x=4,y=4]{v44}
\Vertex[x=5,y=4]{v54}
\Vertex[x=6,y=4]{v64}
\Vertex[x=7,y=4]{v74}
\Vertex[x=8,y=4]{v84}

\Edges(    v10,v20,v30,v40,v50,v60,v70,v80,v90)
\Edges(v01,v11,v21,v31,v41,v51,v61,v71,v81,v91)
\Edges(v02,v12,v22,v32,v42,v52,v62,v72,v82,v92)
\Edges(v03,v13,v23,v33,v43,v53,v63,v73,v83,v93)
\Edges(v04,v14,v24,v34,v44,v54,v64,v74,v84)

\Edge(v01)(v02)
\Edge(v03)(v04)

\Edge(v10)(v11)
\Edge(v12)(v13)

\Edge(v21)(v22)
\Edge(v23)(v24)

\Edge(v30)(v31)
\Edge(v32)(v33)

\Edge(v41)(v42)
\Edge(v43)(v44)

\Edge(v50)(v51)
\Edge(v52)(v53)

\Edge(v61)(v62)
\Edge(v63)(v64)

\Edge(v70)(v71)
\Edge(v72)(v73)

\Edge(v81)(v82)
\Edge(v83)(v84)

\Edge(v90)(v91)
\Edge(v92)(v93)
\end{tikzpicture}
\end{subfigure}
\caption{Walls of height 2, 3, and 4, respectively.}\label{f-walls}
\end{center}
\end{figure}

\noindent
We let ${\cal S}$ be the class of graphs each connected component of which is either a subdivided claw $S_{h,i,j}$ for some $1\leq h\leq i\leq j$ or a path $P_r$ for some $r\geq 1$. 
This leads to the following lemma, which is  well-known and follows from the fact that walls have maximum degree at most~3 and from Lemma~\ref{l-walls} by choosing an appropriate value for $k$  (also note that $k$-subdivided walls are bipartite for all~$k\geq 0$).
 
\begin{lemma}\label{l-classS}
Let $\{H_1,\ldots,H_p\}$ be a finite set of graphs. 
If $H_i\notin {\cal S}$ for $i=1,\ldots,p$ then the class of $\{H_1,\ldots,H_p\}$-free bipartite graphs has unbounded~clique-width.
\end{lemma}

\subsection{A Number of Other Lemmas}

We start with a  lemma which is related to Lemma~\ref{l-observation} and which follows immediately from the corresponding definitions.

\begin{lemma}\label{l-equivalent}
Let $G$ and $H$ be bipartite graphs. 
Then $G$ is $H$-free if and only if $G$ is strongly $H^\ell$-free for all black-and-white labellings $\ell$ of $H$.
\end{lemma}

A graph $G$ that contains a graph $H$ as an induced subgraph may be weakly $H^\ell$-free for all black-and-white labellings $\ell$ of $H$; take for instance the graphs~$G$ and $H$ from the proof of Lemma~\ref{l-observation}. However, we can make the following observation, which also follows directly from the corresponding definitions. 

\begin{lemma}\label{l-equivalent2}
Let $H$ be a bipartite graph with a unique black-and-white labelling $\ell$ (up to isomorphism).
Then every bipartite graph $G$ is $H$-free if and only if it is weakly $H^\ell$-free.
\end{lemma}

Note that there exist both connected bipartite graphs (for example $H=P_6$) and disconnected bipartite graphs (for example $H=2P_2$) 
that satisfy the condition of Lemma~\ref{l-equivalent2}.

Two black-and-white labellings of a bipartite graph $H$ are said to be {\em equivalent} if they are isomorphic or opposite to each other; otherwise they are said to be {\em non-equivalent}. The following lemma follows directly from the definitions. 

\begin{lemma}\label{l-nonequi}
Let $\ell$ and $\ell^*$ be two equivalent black-and-white labellings of a bipartite graph $H$.
Then the class of strongly (weakly) $H^\ell$-free graphs is equal to the class of strongly (weakly) $H^{\ell^*}$-free graphs.
\end{lemma}

The following lemma is due to Lozin and Rautenbach~\cite{LR06}.

\begin{lemma}[\cite{LR06}]\label{l-bipclassS}
Let $\{H_1^{\ell_1},\ldots,H_p^{\ell_p}\}$ be a finite set of labelled bipartite graphs. 
For $i=1,\ldots,p$, let $F_i$ denote the bipartite complement of $H_i$ with respect to $(B^{\ell_i}_{H_i},W^{\ell_i}_{H_i})$.
If $H_i\notin {\cal S}$ for all $1\leq i\leq p$ or $F_i\notin {\cal S}$ for all $1\leq i\leq p$, then the class of strongly $(H_1^{\ell_1},\ldots,H_p^{\ell_p})$-free bipartite graphs 
has unbounded clique-width.
\end{lemma}

In the next lemma we demonstrate a list of $H$-free bipartite classes with unbounded clique-width.
It is obtained by combining a known result of Lozin and Voltz~\cite{LV08} with a number of new results.

\begin{lemma}\label{lem:bip-unbounded}
The class of $H$-free bipartite graphs has unbounded clique-width if $H\in \{2P_1+2P_2, 2P_1+P_4, 4P_1+P_2, 3P_2, 2P_3\}$.
\end{lemma}

\begin{proof}
Lozin and Voltz~\cite{LV08} showed that $2P_3$-free bipartite graphs have unbounded clique-width.
Let $H\in \{2P_1+2P_2, 2P_1+P_4$, $4P_1+P_2, 3P_2\}$, and let $\{H^{\ell_1},\ldots,H^{\ell_p}\}$ be the set of all non-equivalent labelled bipartite graphs isomorphic to $H$.
For $i=1,\ldots,p$, let $F_i$ denote the bipartite complement of $H$ with respect to $(B^{\ell_i}_{H},W^{\ell_i}_{H})$.
We will show that every $F_i$ does not belong to ${\cal S}$.
Then, by Lemma~\ref{l-bipclassS}  the class of strongly $(H_1^{\ell_1},\ldots,H_p^{\ell_p})$-free bipartite graphs has unbounded clique-width.
Because  a bipartite graph is $H$-free if and only if it is strongly $(H_1^{\ell_1},\ldots,H_p^{\ell_p})$-free (by Lemmas~\ref{l-equivalent} and~\ref{l-nonequi}),
 this means that the class of $H$-free bipartite graphs has unbounded clique-width.

Suppose $H\in \{2P_1+2P_2,2P_1+P_4\}$. Let $V_H=\{x_1,\ldots,x_6\}$ with $E_H=\{x_1x_2,x_3x_4\}$ if $H=2P_1+2P_2$ and 
$E_H=\{x_1x_2,x_2x_3,x_3x_4\}$ if $H=2P_1+P_4$. Then $H$ has only two non-equivalent black-and-white labellings. 
We may assume without loss of generality that one of these two labellings colours $x_1,x_3,x_5,x_6$
black and $x_2,x_4$ white, whereas the other one colours $x_1,x_3,x_5$ black and $x_2,x_4$, $x_6$ white.
Let $F_1$ and $F_2$ be the bipartite complements corresponding to the first and second labellings, respectively. 
The vertices $x_2,x_4,x_5,x_6$ induce a $C_4$ in $F_1$, whereas
the vertices $x_1,x_4,x_5,x_6$ induce a $C_4$ in $F_2$. Hence, $F_1$ and $F_2$ do not belong to ${\cal S}$.

Suppose $H=4P_1+P_2$. Let $V_H=\{x_1,\ldots,x_6\}$ and $E_H=\{x_1x_2\}$.
Then $H$ has three non-equivalent black-and-white labellings.
We may assume without loss of generality that the first one colours $x_1,x_3,x_4,x_5,x_6$ black and $x_2$ white, 
the second one colours $x_1,x_3,x_4,x_5$ black and $x_2,x_6$ white, and the third one colours $x_1,x_3,x_4$ black and $x_2,x_5,x_6$ white.
Let $F_1,F_2,F_3$ denote the corresponding bipartite complements.
The vertices $x_2,\ldots,x_6$ induce a $K_{1,4}$ in $F_1$.
The vertices $x_2,x_3,x_4,x_6$ induce a $C_4$ in $F_2$ and $F_3$.
Hence, none of $F_1,F_2,F_3$ belongs to ${\cal S}$.

Suppose $H=3P_2$. Let $V_H=\{x_1,\ldots,x_6\}$ and $E_H=\{x_1x_2,x_3x_4,x_5x_6\}$. Let $\ell$ be a black-and-white labelling of $H$ that colours $x_1,x_3,x_5$ black
and $x_2,x_4,x_6$ white. Then any any other labelling $\ell^*$ of $H$ is isomorphic to $\ell$. The bipartite complement of $H$ with respect to $(B_H^\ell,W_H^\ell)$ is isomorphic to $C_6$, which does not belong to ${\cal S}$.
\qed
\end{proof}

We will also need the following lemma.

\begin{lemma}\label{l-addit}
 Let $H\in {\cal S}$. Then $H$ is 
 $(2P_1+2P_2, 2P_1+P_4, 4P_1+P_2,3P_2, 2P_3$)-free
if and only if $H=sP_1$ for some integer $s\geq 1$ or $H$ is an induced subgraph of one of the graphs in 
$\{K_{1,3}+3P_1,K_{1,3}+P_2, P_1+S_{1,1,3}, S_{1,2,3}\}$.
\end{lemma}

\begin{proof}
Let $H\in {\cal S}$.
First suppose that $H=sP_1$ for some integer $s\geq 1$ or that $H$ is an induced subgraph of one of the graphs in 
$\{K_{1,3}+3P_1,K_{1,3}+P_2, P_1+S_{1,1,3}, S_{1,2,3}\}$.
It is readily seen that $H$ is  $(2P_1+2P_2, 2P_1+P_4, 4P_1+P_2,3P_2, 2P_3$)-free.

Now suppose that $H$ is  $(2P_1+2P_2, 2P_1+P_4, 4P_1+P_2,3P_2, 2P_3$)-free.
First suppose that $H$ has at least five connected components.
Because $H$ is $(4P_1+P_2)$-free, we find that $H=sP_1$ for some integer $s\geq 5$.

Now suppose that $H$ has exactly four connected components. 
Because $H$ is $(2P_1+2P_2)$-free, $H=3P_1+D$, where $D$ may have more than one edge. 
Because $H$ is $(2P_1+2P_2,2P_1+P_4)$-free, $D$ is $(2P_2,P_4)$-free.
As $H\in {\cal S}$, this means that $D$ is isomorphic to one of $\{K_{1,3},P_1,P_2,P_3\}$.
Hence, $H$ is an induced subgraph of $K_{1,3}+3P_1$.

Now suppose that $H$ has exactly three connected components, say $H=D_1+D_2+D_3$ with $|V_{D_1}|\leq |V_{D_2}| \leq |V_{D_3}|$.
Because $H$ is $3P_2$-free, we may assume without loss of generality that $D_1=P_1$.
First suppose that $D_3$ is $P_3$-free. Then, as $H\in {\cal S}$, we find that $H$ is an induced subgraph of $P_1+2P_2$, which is an induced subgraph of $S_{1,2,3}$.
Now suppose that $D_3$ is not $P_3$-free. 
Because $H=(2P_1+P_4)$-free, $D_3$ is $P_4$-free.
As $H\in {\cal S}$, this means that $D_3\in \{K_{1,3},P_3\}$.
Moreover, as  $H$ is $2P_3$-free and $H\in {\cal S}$, we find that
$D_2\in \{P_1,P_2\}$. Because $H$ is $(4P_1+P_2)$-free, the combination $D_2=P_2$ and $D_3=K_{1,3}$ is not possible.
Hence, $H$ belongs to $\{K_{1,3}+2P_1, 2P_1+P_3,P_1+P_2+P_3\}$, which means that $H$ is an induced subgraph of
$K_{1,3}+3P_1$ or of $S_{1,2,3}$.

Now suppose that $H$ has exactly two connected components, say $H=D_1+D_2$ with $|V_{D_1}|\leq |V_{D_2}|$.
First suppose that $D_2$ is $P_3$-free. 
Then, as $H\in {\cal S}$, we find that $H$ is an induced subgraph of $2P_2$, which is an induced subgraph of $S_{1,2,3}$.
Now suppose that $D_2$ is not $P_3$-free.
Because $H$ is $2P_3$-free and $H\in {\cal S}$, we find that $D_1\in \{P_1,P_2\}$ and that $D_2$ is either a path or a subdivided claw.
Because $H$ is $(2P_1+P_4)$-free, $D_2$ is $P_6$-free. 
Suppose that $D_2$ is a path. Then $D_2\ssi P_5$. If $D_2=P_5$ then $D_1=P_1$, as $H$ is $3P_2$-free. 
Hence, we find that $H\in \{P_1+P_3,P_1+P_4,P_1+P_5,P_2+P_3,P_2+P_4\}$, which means that $H$ is an induced subgraph of $P_1+S_{1,1,3}$ or of $S_{1,2,3}$.
Suppose that $D_2$ is a subdivided claw, say $D_2=S_{a,b,c}$ for some $1\leq a\leq b\leq c$. Then, because $H$ is $(2P_1+2P_2)$-free, $a=b=1$.
Because $H$ is $(2P_1+P_4)$-free, $c\leq 3$. Moreover, if $2\leq c\leq 3$ then $D_1=P_1$, as $H=(2P_1+2P_2)$-free. Hence, we find that 
$H\in \{K_{1,3}+P_1,K_{1,3}+P_2,P_1+S_{1,1,2}, P_1+S_{1,1,3}\}$,
which means that $H$ is an induced subgraph of $K_{1,3}+P_2$ or of $P_1+S_{1,1,3}$.

Now suppose that $H$ has exactly one connected component. As $H\in {\cal S}$, we find that $H$ is either a path or a subdivided claw.
If $H$ is a path then, as $H$ is $2P_3$-free, $H$ is an induced subgraph of $P_6$, which means that $H\ssi S_{1,2,3}$.
Suppose that $H$ is a subdivided claw, say $H=S_{a,b,c}$ for some $1\leq a\leq b\leq c$.
Because $H$ is $3P_2$-free, we find that $a=1$.
Because $H$ is $2P_3$-free, we find that $b\leq 2$ and that $c\leq 3$.
 Hence, $H$ is an induced subgraph of $S_{1,2,3}$.
This completes the proof.\qed
\end{proof}

The last lemma we need before proving the main results of this paper is the following one (we use it several times in the proof of Theorem~\ref{thm:bipweak}). 

\begin{lemma}\label{l-last}
Let $H^\ell$ be a labelled bipartite graph. 
The class of weakly $H^\ell$-free bipartite graphs has unbounded clique-width in both of the following cases:
\begin{itemize}
\item [(i)] $H^\ell$ contains a vertex of degree at least~$3$, or
\item [(ii)] $H^\ell$ contains four independent vertices, not all of the same colour.
\end{itemize}
\end{lemma}

\begin{proof}

Let $b_1$ be a black-and-white labelling of $4P_1$ that colours
three vertices black and one vertex white.
Let $b_2$ be a black-and-white labelling of $4P_1$
that colours two vertices black and two vertices white.
We show below that the class of weakly $H^\ell$-free bipartite graphs has unbounded clique-width if $H^\ell\in \{(K_{1,3})^b,(4P_1)^{b_2},(4P_1)^{b_3}\}$.
Then we are done by Lemma~\ref{l-nonequi}.

Consider a 1-subdivided wall $G'$ obtained from a wall $G$. Recall that 1-subdivided walls  are bipartite. Moreover, the vertices that were introduced when subdividing every edge 
of $G$ all  have degree~2 and form one class of a bipartition  $(B,W)$ of $G'$. Let this class be $B$. Then $(K_{1,3})^b$ is not a labelled induced subgraph of $(B,W,E_{G'})$.
Hence, $G'$ is weakly $(K_{1,3})^b$-free. This means that the class of weakly $(K_{1,3})^b$-free graphs contains the class of 1-subdivided walls. As such, it has unbounded clique-width 
by Lemma~\ref{l-walls}. 
The bipartite complement $G''$ of $G'$ with respect to $(B,W)$ is weakly $(4P_1)^{b_1}$-free, as $(K_{1,3})^{b}$ is the bipartite complement of $(4P_1)^{b_1}$ and $(K_{1,3})^{b}$ is not a labelled induced subgraph of $(B,W,E_{G'})$.
Hence, the class of weakly $(4P_1)^{b_1}$-free graphs has unbounded clique-width by Fact~2.
The class of weakly $(4P_1)^{b_2}$-free bipartite graphs has unbounded clique-width by Lemma~\ref{l-observation} and Theorem~\ref{thm:bipvadim}. 
\qed
\end{proof}

\subsection{The Proof of Theorem~\ref{thm:bipartite}}

\begin{proof}
We first deal with the bounded cases. 
First suppose $H=sP_1$ for some $s\geq 1$. Then any
$H$-free bipartite graph $G$ can have at most $s-1$ vertices in each partition class of any bipartition.
This means that the clique-width of $G$ is at most $2s-2$.  
Now suppose that $H\in 
\{K_{1,3}+3P_1,K_{1,3}+P_2, P_1+S_{1,1,3},S_{1,2,3}\}$. Then the claim follows from combining Lemma~\ref{l-observation} with Theorem~\ref{thm:bipvadim}.

We now deal with the unbounded cases.
Suppose  $H\neq sP_1$ for any $s \geq 1$ and that $H$ is not an induced subgraph of one of the graphs in 
$\{K_{1,3}+3P_1,K_{1,3}+P_2, P_1+S_{1,1,3}, S_{1,2,3}\}$.
Then by Lemma~\ref{l-addit}, either $H\notin {\cal S}$ or, $H$ is not
 $(2P_1+2P_2, 2P_1+P_4, 4P_1+P_2,3P_2, 2P_3$)-free.
Hence, the clique-width of the class of $H$-free bipartite graphs is unbounded by 
Lemmas~\ref{l-classS} and~\ref{lem:bip-unbounded}, respectively.\qed
\end{proof}

\subsection{The Proof of Theorem~\ref{thm:bipweak}}

\begin{proof}
We first consider the bounded cases.
First suppose $H^\ell=(sP_1)^b$ for some $s\geq 1$ (the $H^\ell=(sP_1)^{\overline{b}}$ case is equivalent). Then any weakly $H^\ell$-free bipartite graph has a bipartition $(B,W)$ with $|B|\leq s-1$.
Hence, the clique-width of such a graph is at most $s+1$ (first introduce the vertices of $B$ by using distinct labels, then use two more labels for the vertices of $W$, introducing them one-by-one).

Now suppose $H^\ell=(2P_1+P_3)^b$. Let $G$ be a weakly $H^\ell$-free bipartite graph.
Then $G$ has a bipartition $(B,W)$ such that $H^\ell$ is not a labelled induced subgraph of $(B,W,E_G)$.
By the previous paragraph, we may assume without loss of generality that $|B|\geq 4$ and $|W|\geq 1$.
Then every white vertex in $G$ is adjacent to either all but at most one black vertex or
non-adjacent to all but at most one black vertex. Let $W'$ be the set of white
vertices adjacent to all but at most one black vertex.
Apply a bipartite complementation between $W'$ and $B$.
The resulting graph is a disjoint union of stars, which have clique-width at most~2. Thus, by Fact~2, the class of weakly $H^\ell$-free
graphs has bounded clique-width. 

Before considering the case $H^\ell=(P_1+P_5)^b$, we first consider the case where
$H\ssi P_2+P_4$ or $H\ssi P_6$. We first assume that $H=P_2+P_4$ or $H=P_6$.
Then $H\ssi S_{1,2,3}$, which implies that that the class of $H$-free bipartite graphs has bounded clique-width by Theorem~\ref{thm:bipartite}.
All black-and-white labellings of $P_2+P_4$ are isomorphic.
Similarly, all black-and-white labellings of $P_6$ are isomorphic.
Hence, the class of $H$-free bipartite graphs coincides with the class of
weakly $H^\ell$-free graphs by Lemma~\ref{l-equivalent2}. We therefore conclude that the
latter class also has bounded clique-width.

Now let $H \ssi P_2+P_4$ or $H\ssi P_6$, but $H \not \in \{P_2+P_4,P_6\}$. 
Note that $P_2+P_4$ and $P_6$ have a unique labelling $b$ (up to isomorphism).
If $H^\ell$ is not a labelled induced subgraph of one of 
$\{(P_2+P_4)^b,P_6^b\}$
then $H$ must have two non-equivalent black-and-white
labellings. Since $H$ is a linear forest, it must have at least two components with an
odd number of vertices. Therefore $H \in \{2P_1,3P_1,P_1+P_3,2P_1+P_2\}$.
However, in all these cases, for any labelling~$\ell$ of $H$, $H^\ell \li P_6^b$
or $H^\ell\li (P_2+P_4)^b$. Therefore, if $H\ssi P_2+P_4$ or $H\ssi P_6$ then for any
labelling $\ell$ of $H$, the weakly $H^\ell$-free bipartite graphs are a subclass
of either the $P_6$-free or $(P_2+P_4)$-free bipartite graphs. In particular, this holds for $H^\ell=(P_1+2P_2)^b$ (we need
this observation for the following case).

Finally, suppose $H^\ell=(P_1+P_5)^b$. 
Let $G$ be a weakly $H^\ell$-free bipartite graph.
Then $G$ has a labelling $\ell^*$ such that $H^\ell$ is not a labelled induced subgraph of $(B_G^{\ell^*},W_G^{\ell^*},E_G)$.
If $|B_G^{\ell^*}|$ is even, then we delete a vertex of $B_G^{\ell^*}$. We may do this by Fact~1.
Hence $|B_G^{\ell^*}|$ may be assumed to be odd.
Let $X$ be the subset of $W_G^{\ell^*}$ that consists of all vertices that are  adjacent to less than half of the vertices of $B_G^{\ell^*}$.
We apply the bipartite complementation between $X$ and $B_G^{\ell^*}$. We may do this by Fact~2.
Let $G_1$ be the resulting bipartite graph, with bipartition classes $B_{G_1}^{\ell^*}=B_G^{\ell^*}$ and $W_{G_1}^{\ell^*}=W_G^{\ell^*}$.

Suppose $B_{G_1}^{\ell^*}$ contains three vertices $b_1,b_2,b_3$ and $W_{G_1}^{\ell^*}$ contains two vertices $w_1,w_2$ such that $G_1^{\ell^*}[b_1,b_2,b_3,w_1,w_2]$ is isomorphic to $(P_1+2P_2)^b$. 
By construction and because $|B_{G_1}^{\ell^*}|=|B_G^{\ell^*}|$ is odd, $w_1$ and $w_2$ have at least one common neighbour $b_4\in B_{G_1}^{\ell^*}$. 
Then $G_1^{\ell^*}[b_1,b_2,b_3,b_4,w_1,w_2]$ is isomorphic to $(P_1+P_5)^b$. However, then $G^{\ell^*}[b_1,b_2,b_3,b_4,w_1,w_2]$ is also isomorphic to $(P_1+P_5)^b$ (irrespective of whether $w_1$ or $w_2$ belong to $X$), which is a contradiction. 
We conclude that $G_1$ is weakly $(P_1+2P_2)^b$-free. 
As observed above, this means that $G_1$ has bounded clique-width. Hence $G$ has bounded clique-width.

\medskip
\noindent
We now consider the unbounded cases. 
Let $H^\ell$
be a labelled bipartite graph that is not isomorphic to one of the (bounded) cases considered already.
Suppose that $H$ contains a cycle or an induced subgraph isomorphic to $2P_3$.
Then the class of weakly $H^\ell$-free graphs has unbounded clique-width by combining Lemma~\ref{l-observation} with Theorem~\ref{thm:bipartite}.
Suppose that $H$ contains a vertex of degree at least~3.
Then the class of weakly $H^\ell$-free bipartite graphs has unbounded clique-width by Lemma~\ref{l-last}(i).
It remains to consider the case when $H=sP_1+tP_2+P_r$ for some constants $1\leq r\leq 6$, $s\geq 0$ and $t\geq 0$, where $\max\{s,t\}\geq 1$ (as $H$
is not an induced subgraph of $P_6$).

Suppose $5\leq r\leq 6$. Assume without loss of generality that three vertices of the copy of $P_r$ in $H^\ell$ are coloured black. If $r=6$ or $t\geq 1$ or some copy~$P_1$ in $H^\ell$ is coloured white, or two copies of $P_1$ in $H^\ell$ are coloured black, then we can apply Lemma~\ref{l-last}(ii). Hence, $H^\ell=(P_1+P_5)^b$, which is not possible by assumption.
 
Suppose $r=4$. If two vertices in the induced subgraph of $H^\ell$ isomorphic to $sP_1+tP_2$ have the same colour then we can apply Lemma~\ref{l-last}(ii). 
Hence we may assume that $s\leq 2$ and $t\leq 1$, and moreover that $s=0$ if $t=1$. Also we would have $H\ssi P_2+P_4$ if $s=0$ and $t=1$ or if $s=1$ and $t=0$.
Hence, it remains to consider the case $s=2$ and $t=0$, such that one copy of $P_1$ is coloured black and the other one white.
In that case, we may apply Lemma~\ref{l-last}(ii).

Suppose  $r=3$. Assume without loss of generality that the two vertices of the copy of $P_3$ in $H^\ell$ are coloured black.
Recall that $s\geq 1$ or $t\geq 1$. If $t\geq 2$, then we can apply Lemma~\ref{l-last}(ii). Suppose $t=1$.
The $s=0$ otherwise $H^\ell$ would contain an induced $4P_1$ in which not all the vertices are the same colour, in which case we could apply Lemma~\ref{l-last}(ii).
However, this means that $H$ is an induced subgraph of $P_2+P_4$. 
Now suppose $t=0$. Then $s\geq 2$, as otherwise $H$ is an induced subgraph of $P_2+P_4$.
If $s\geq 3$ then $H^\ell$ contains an induced $4P_1$ in which not all the vertices are the same
 colour, in which case we apply Lemma~\ref{l-last}(ii).
Hence, $s=2$ and both copies are coloured black (otherwise we apply Lemma~\ref{l-last}(ii)). However, in this case $H^\ell$ is isomorphic to $(2P_1+P_3)^b$, which is not possible by assumption.

Finally suppose that $r\leq 2$. Then we may write $H=sP_1+tP_2$ instead.
We must have $s+t \geq 4$ or $t \geq 3$, otherwise $H$ would be an induced subgraph of $P_2+P_4$ or $P_6$. If $t=0$ then since $H^\ell\neq (sP_1)^b$ and $H^\ell\neq (sP_1)^{\overline{b}}$ we can find four copies of $P_1$ in $H$ that are not all of the same colour and apply Lemma~\ref{l-last}(ii). If $t\geq1, s+t \geq 4$, we can also find four copies of $P_1$ that are not all of the same colour and apply Lemma~\ref{l-last}(ii). 
Finally, suppose $s=0,t=3$. In this case we combine Lemmas~\ref{l-observation} and~\ref{lem:bip-unbounded}.
This completes the proof.\qed
\end{proof}

\section{Conclusions}\label{s-con}

We have completely determined those bipartite graphs $H$ for which the class of $H$-free bipartite graphs has bounded clique-width.
We also characterized exactly those labelled bipartite graphs $H$ for which the class of 
weakly $H$-free bipartite graphs has bounded clique-width. These results complement the known characterization of Lozin and Volz~\cite{LV08}
for strongly $H$-free bipartite graphs. 
A natural direction for further research would be to characterize, for each of the three notions of $H$-freeness, 
the clique-width of classes of ${\cal H}$-free bipartite graphs when ${\cal H}$ is a set containing at least~2 graphs.
 In a follow-up paper~\cite{DP}, we apply our results for $H$-free bipartite graphs to determine classes of $(H_1,H_2)$-free
(general) graphs of bounded and unbounded clique-width. 

\bibliographystyle{abbrv}

\bibliography{mybib}

\begin{thebibliography}{10}

\bibitem{BL02}
R.~Boliac and V.~Lozin.
\newblock On the clique-width of graphs in hereditary classes.
\newblock {\em Lecture Notes in Computer Science}, 2518:44--54, 2002.

\bibitem{BGMS14}
F.~Bonomo, L.~N. Grippo, M.~Milani\v{c}, and M.~D. Safe.
\newblock Graphs of power-bounded clique-width.
\newblock {\em arXiv}, abs/1402.2135, 2014.

\bibitem{BELL06}
A.~Brandstadt, J.~Engelfriet, H.-O. Le, and V.~Lozin.
\newblock Clique-width for 4-vertex forbidden subgraphs.
\newblock {\em Theory of Computing Systems}, 39(4):561--590, 2006.

\bibitem{BKM06}
A.~Brandst{\"a}dt, T.~Klembt, and S.~Mahfud.
\newblock {$P_6$}- and triangle-free graphs revisited: structure and bounded
  clique-width.
\newblock {\em Discrete Mathematics and Theoretical Computer Science},
  8(1):173--188, 2006.

\bibitem{BK05}
A.~Brandstädt and D.~Kratsch.
\newblock On the structure of ({$P_5$},gem)-free graphs.
\newblock {\em Discrete Applied Mathematics}, 145(2):155--166, 2005.

\bibitem{BLM04b}
A.~Brandstädt, H.-O. Le, and R.~Mosca.
\newblock Gem- and co-gem-free graphs have bounded clique-width.
\newblock {\em International Journal of Foundations of Computer Science},
  15(1):163--185, 2004.

\bibitem{BLM04}
A.~Brandstädt, H.-O. Le, and R.~Mosca.
\newblock Chordal co-gem-free and ({$P_5$},gem)-free graphs have bounded
  clique-width.
\newblock {\em Discrete Applied Mathematics}, 145(2):232--241, 2005.

\bibitem{BM02}
A.~Brandstädt and S.~Mahfud.
\newblock Maximum weight stable set on graphs without claw and co-claw (and
  similar graph classes) can be solved in linear time.
\newblock {\em Information Processing Letters}, 84(5):251--259, 2002.

\bibitem{BM03}
A.~Brandstädt and R.~Mosca.
\newblock On variations of {$P_4$}-sparse graphs.
\newblock {\em Discrete Applied Mathematics}, 129(2--3):521--532, 2003.

\bibitem{CMR00}
B.~Courcelle, J.~A. Makowsky, and U.~Rotics.
\newblock Linear time solvable optimization problems on graphs of bounded
  clique-width.
\newblock {\em Theory of Computing Systems}, 33(2):125--150, 2000.

\bibitem{DGP13}
K.~K. Dabrowski, P.~A. Golovach, and D.~Paulusma.
\newblock Colouring of graphs with {R}amsey-type forbidden subgraphs.
\newblock {\em Theoretical Computer Science}, 522:34--43, 2013.

\bibitem{DP}
K.~K. Dabrowski and D.~Paulusma.
\newblock Clique-width of graph classes defined by two forbidden induced
  subgraphs.
\newblock {\em Manuscript}, 2014.

\bibitem{Gu07}
F.~Gurski.
\newblock Graph operations on clique-width bounded graphs.
\newblock {\em CoRR}, abs/cs/0701185, 2007.

\bibitem{KLM09}
M.~Kamiński, V.~Lozin, and M.~Milanič.
\newblock Recent developments on graphs of bounded clique-width.
\newblock {\em Discrete Applied Mathematics}, 157(12):2747--2761, 2009.

\bibitem{KR03b}
D.~Kobler and U.~Rotics.
\newblock Edge dominating set and colorings on graphs with fixed clique-width.
\newblock {\em Discrete Applied Mathematics}, 126(2--3):197--221, 2003.

\bibitem{LR04}
V.~Lozin and D.~Rautenbach.
\newblock On the band-, tree-, and clique-width of graphs with bounded vertex
  degree.
\newblock {\em SIAM Journal on Discrete Mathematics}, 18(1):195--206, 2004.

\bibitem{LR06}
V.~Lozin and D.~Rautenbach.
\newblock The tree- and clique-width of bipartite graphs in special classes.
\newblock {\em Australasian Journal of Combinatorics}, 34:57--67, 2006.

\bibitem{LV08}
V.~Lozin and J.~Volz.
\newblock The clique-width of bipartite graphs in monogenic classes.
\newblock {\em International Journal of Foundations of Computer Science},
  19(02):477--494, 2008.

\bibitem{MR99}
J.~Makowsky and U.~Rotics.
\newblock On the clique-width of graphs with few {$P_4$}'s.
\newblock {\em Int. J. Found. Comput. Sci.}, 10(3):329--348, 1999.

\bibitem{Oum08}
S.-I. Oum.
\newblock Approximating rank-width and clique-width quickly.
\newblock {\em ACM Trans. Algorithms}, 5(1):10:1--10:20, 2008.

\bibitem{Ra07}
M.~Rao.
\newblock {MSOL} partitioning problems on graphs of bounded treewidth and
  clique-width.
\newblock {\em Theoretical Computer Science}, 377(1--3):260--267, 2007.

\end{thebibliography}

\end{document}